\documentclass{article}

\usepackage{todonotes}
\usepackage{fullpage}
\usepackage{amsfonts}
\usepackage{graphicx}
\usepackage{amsmath}
\usepackage[english]{babel}
\usepackage[colorlinks,bookmarks=false,linkcolor=blue,urlcolor=blue]{hyperref}
\usepackage{amsthm}
\usepackage{verbatim}
\usepackage{enumerate}
\usepackage{natbib}

\usepackage{tikz}
\usetikzlibrary{snakes}
\usepackage{float}
\usepackage{color}
\usepackage{algorithm}
\usepackage{algpseudocode}
\usepackage{caption}

\newtheorem{theorem}{Theorem}[section]

\newtheorem{corollary}[theorem]{Corollary}

\newtheorem{definition}[theorem]{Definition}

\newtheorem{proposition}[theorem]{Proposition}

\theoremstyle{definition}
\newtheorem*{example}{Example}

\begin{document}
\title{Popularity in location games}
\author{Ga\"etan Fournier\footnote{Aix Marseille Univ, CNRS, AMSE, Marseille, France. The project leading to this publication has received funding from the french government under the “France 2030” investment plan managed by the French National Research Agency (reference :ANR-17-EURE-0020) and from Excellence Initiative of Aix-Marseille University - A*MIDEX. " 
} \and Marc Schr\"oder\footnote{Maastricht University, Maastricht, Netherlands}}
\maketitle

\begin{abstract}
We study a variant of the Hotelling-Downs model of spatial competition between firms where consumer choices are influenced by their individual preferences as well as the popularity of the firms. In general, a multiplicity of market equilibria might exist due to the popularity effect. To elucidate firm decision-making, we explore three distinct behavioral attitudes —optimistic, neutral, and pessimistic— towards this multiplicity of equilibria.  For each behavior, we characterize the set of Nash equilibria and measure the impact of the selfish behavior on the social welfare by means of the price of anarchy and price of stability.
\end{abstract}


\section{Introduction}

In their decision making processes, economic agents are often influenced by the choices of others. It is for example well documented that more demanded goods tend to be more desirable by customers and that the popularity of a candidate serves as a powerful asset in persuading new voters\footnote{We refer to \cite{nadeau1993new} and \cite{dizney1962investigation} for the two mentioned examples.}. The bandwagon effect, herding and conformism are societal manifestations of the positive externalities arising from the majority's choices on individuals. However, these effects can be in contradiction with the individual preferences of agents. The current paper aims to explore their consequences on a Hotelling-Downs competitive environment.\\
~~\\
To do so, we first investigate how economic agents trade-off between (1) taking an action in line with their idiosyncratic preferences and (2) aiming to select a popular decision. We then study how competitors, such as firms or political parties, anticipate this trade-off when selecting their locations in the first place. This location can be interpreted as the characteristics of a good to produce or the political platform to announce for an election, and we study how it depends on the popularity effect at equilibrium.\\
~~\\
\underline{First motivating example:} Consider a market where two firms want to maximize their clientele by selling a product whose price is exogenously determined\footnote{This model applies for example to newsstands, pharmacies, or franchises of different types of services and products. It also applies to the media markets or to any sector where consumers are not charged with direct prices but through advertising.}.  We suppose that the good has a $1$-dimensional characteristic and that, in the absence of externality, consumers would choose to buy one unit of the product whose characteristic is the closest to their ideal product. The contribution of the current paper is to add, in the consumers' preferences a term that represents their inclination towards the most popular good, that we model with the addition of a linear function of the market size of the firm that sells this good. Several arguments support the existence of such a term, let us mention the social proof: popularity being perceived as a signal of quality when agents have asymmetric information (see \cite{mutz1998impersonal} for example\footnote{In \cite{mutz1998impersonal}: "consumers who need a new car but lack the expertise to judge which one is best, they take majority opinion as a proxy for the most intelligent choice: If so many people are in favor of a particular alternative, there must be something to it, and it can therefore safely be chosen without further detailed knowledge."}) and the need for conformism of economic agents (see \cite{corneo1997conspicuous}).\\
~~\\
\underline{Second motivating example:} Consider an election where two candidates select platforms $x_1,x_2 \in [0,1]$ to attract as many voters as possible\footnote{A proportional representative election, or a winner-takes-all election where voters have incomplete information about political parties, see Lindbeck and Weibull (1987), or where parties have incomplete information about the distribution of voters’ preferences Patty (2002).}. On top of considering the distance between their ideal policy and the proposed platforms, the electorate also observes poll reports that influence their decision. Their effect on voter preferences is well understood, see for example \cite{mehrabian1998effects}, and the underlying mechanisms include the desire to be on the winning side, the apprehension of casting a futile vote, and the heightened media exposure of popular candidates.\\
~~\\
Stock market frenzies\footnote{See \cite{spyrou2013herding} for a review of herding in financial markets.} or fashion trends are other manifestations of the positive externalities of popular choices on consumers' decision-making. For simplicity, in the sequel of the paper, we use the terminology of the first example: players are firms that aim at maximizing their clientele.\\
~~\\
\textbf{Our model.} Two firms simultaneously select the characteristics of the product they sell $x_1,x_2 \in [0,1]$.  Once locations are selected, we study the choices of consumers, who 
buy from the firm that maximizes their utility. A consumer $t \in [0,1]$, identified with its preferred characteristics, has a utility to buy good $i\in \{1,2\}$ which is a sum of two terms. The first one is a cost that we take equal to the distance between the product characteristic $x_i$ and the ideal of the consumer $t$, which is standard in location games. The second term is an increasing linear function of the market size of good $i$, the linear coefficient being a proxy for the magnitude of the popularity effect. Determining the decision of a consumer requires computing the market shares that in turn requires knowing the decision of consumers. Solving this equation leads to the computation of an equilibrium that we call the market equilibrium. Such a market equilibrium is guaranteed to exist for every pair of locations.\\
~~\\
\textbf{Our results.} Following the spirit of backward induction, we start by considering the choice of the consumers for fixed characteristics $x_1,x_2$. We find that there exists, in general, a multiplicity of market equilibria. We describe up to $5$ possible outcomes: firm $1$ or $2$ covers the full market, they split the market with a unique indifferent consumer located between $x_1$ and $x_2$, or they split the market with all consumers on the left to $x_1$, respectively to the right to $x_2$, being indifferent. Obviously, among these different equilibria, some favor one firm or the other, and the anticipation of the market reaction is key to select the ex-ante locations $x_1,x_2$. Therefore, we introduce several possible behaviors for firms, depending on which market equilibria they trust to be realized after the choice of their locations. Namely, we consider "pessimistic", "neutral" or "optimistic" firms, depending on their beliefs regarding the market reaction when a firm deviates from the status-quo. A pessimistic (resp. optimistic) firm anticipates the worst (resp. best) market equilibrium to appear, that is the equilibrium that minimizes (resp. maximizes) its market share. On the other hand, a neutral firm computes its expected market share when all equilibria have the same probability.\\
~~\\
We find different equilibrium structures dependent on the behavior of firms: there is no equilibrium when firms are optimistic, as their optimism makes deviation profitable in general. When firms are neutral, we find that there exists an equilibrium if and only if the magnitude $a$ of the positive externality is smaller than $\frac12$ (which is the marginal gain for selecting a location closer to the opponent, in the absence of popularity effect). In this case, the effects of the competition over characteristics dominate the effects of the competition over popularity and we find that the principle of minimal differentiation applies: both firms select the same central location, as in the standard Hotelling-Downs model.\\
~~\\
We find a richer equilibrium structure for pessimistic firms. In this case, we find potential differentiated equilibrium: different characteristics are selected in equilibrium. The resulting differentiation is explained by the fact that firms prefer to cultivate separate customers' bases and not risk losing all their market shares by competing too strongly with their competitors. We indeed find that the positive externality of popularity prevents a pessimistic firm from deviating too close to its opponent, which is the strong incentive that generates a convergence of firms in the standard Hotelling-Downs model. In the case where the two products are too similar, the market equilibria include situations where pessimistic firms anticipate losing their entire clientele.\\
~~\\
Finally, we investigate the welfare consequences of the popularity effects. A term like herding has a negative connotation and is generally negatively perceived. We are not aware of any previous analysis of its impact on welfare with strategic firms. In our model, we observe two contradictory impacts: while efficient coordination between agents can lead to an important surplus, the effect could also have negative consequences, such as reducing the diversity of offered options, which in turn impact negatively the welfare of consumers. Interestingly, our welfare analysis differs from the standard discussion in industrial organization or in political economy, where the inefficiency of equilibrium relies specifically on firms or parties to differentiate too much (polarization) or not enough (homogenization). The popularity effect modifies this perspective: when $a$ is large enough for example, the social optimum makes consumers indifferent about the differentiation of the two firms, they both buy from a popular firm even if the product differs significantly from their ideal.\\ 

The total firms' profit being constantly equal to 1, we can restrict the welfare analysis to the consumers' side without loss of generality. On the firm's side, we only investigate asymmetry in payoffs and we find that this asymmetry can be greater with larger $a$, reaching the most extreme situation when $a \rightarrow \frac12$. On the consumer's side, we compute the optimal configuration $(x_1^*,x_2^*)$ that maximizes the consumers' surplus. Together with the equilibrium characterization, we can derive the price of anarchy and the price of stability, namely the ratio between the worst (resp. best) equilibrium outcome and the social optimum. Surprisingly, we find a non-monotonic effect of the magnitude of the positive externality, measured by the parameter $a$: in the case of neutral firms, they are both decreasing then increasing, reaching their minimum when $a=\frac14$. In the case of pessimistic firms, the price of anarchy is first decreasing then increasing with $a$, reaching its minimum when $a=\frac{\sqrt{57}-7}{4}\simeq 0.137$ \footnote{In the case where the willingness to buy $\theta$ is set to $1$. We show more generally that the minimum is reached when $a=\frac{1-8\theta+\sqrt{64 \theta ^2 - 16 \theta + 9}}{4}$.}. The price of stability is not monotonic either, before reaching $1$ for $a \geq \frac12$.\\
~~\\
\textbf{Related literature.}

The trend of economic agents preferring popular decisions was introduced by the seminal paper \cite{asch2016effects}, and was soon confirmed by many empirical and experimental evidence. We refer to \cite{marsh1985back} and \cite{nadeau1993new} for summaries of the previous research of empirical evidence in political economy. A recent survey of the experimental literature can be found in \cite{fatas2018preference}. The underlying mechanisms of the preference for popular decisions include the fact that actions can be a signal of quality in the presence of uncertainty (\cite{banerjee1992simple}; \cite{bikhchandani1992theory}), the fact that similar actions might create mutual positive externalities (see, e.g., \cite{katz1986technology}; \cite{banerjee1990peer}) or the fact that individuals have an inherent tendency to identify with a certain class of people (\cite{bernheim1994theory}; \cite{akerlof1980theory}).\\

In our paper, consumers have a preference for popular goods, and firms strategically anticipate this preference when selecting the characteristics of the good they produce. We therefore contribute to the literature concerned with the differentiation of products. In location games, the differentiation of firms is often explained by the existence of a second stage, for example, a price competition (\cite{d1979hotelling}), a potential entry of a third firm (\cite{palfrey1984spatial}) or an oligopoly (\cite{eaton1975principle}). In our paper, firms only play once and no entry is possible, however, pessimistic firms might differentiate at equilibrium to secure their clientele. Indeed, when deviating closer to their opponent, firms create a possible threat where the market equilibrium is unfavorable.\\
~~\\
A similar location game is studied in \cite{kohlberg1983equilibrium,peters2018hotelling,feldotto2019hotelling}, where the authors consider negative externalities. In this setting the market reaction is unique \cite{kohlberg1983equilibrium} and based on this assumption, \cite{peters2018hotelling} was able to characterize equilibria for a small, but even number of firms whenever the externality effect is strong enough. Equilibria never exist if we have an odd and small number of firms. Note that negative externalities have a very different impact compared to positive ones. Indeed, negative externalities introduce a moderating effect on popularity, whereas positive externalities have an amplifying effect. This amplification leads to a multiplicity of market equilibria.\\
~~\\
The market game played between consumers without location effects corresponds to a non-atomic congestion game with linearly decreasing cost functions. The common solution concept for these games is the Wardrop equilibrium, which is known to exist due to \cite{beckmann1956studies}. The inefficiency of Wardrop equilibria was first analyzed by \cite{roughgarden2002bad} and was later characterized by means of simpler proofs in \cite{correa2008geometric}. However, most of these inefficiency bounds apply to increasing cost functions. An important other class of non-atomic games with decreasing cost functions are networks design games, see e.g. \cite{anshelevich2008price}.  In contrast to non-atomic games in which the identity of players is irrelevant, we assume that our consumers are uniformly spread over the unit interval which has some consequences for the market equilibria that we study.\\
~~\\
We also contribute to the literature on political economy that study the bandwagon effect on elections (see \cite{barnfield2020think} for a recent survey on this topic). The literature that analyzes the effects of public opinion polls prior is surveyed in \cite{irwin2000bandwagons}. Experimental studies focused on the electoral effect of publishing opinion polls and highlighted bandwagon effects (\cite{klor2007welfare}, \cite{grosser2010public}). A large part of the literature focused on the effect of the voting system on the bandwagon effect (\cite{myerson1993theory}) and its relationship with the electorate participation (\cite{grillo2017risk}). We investigate in the current paper the relationship between the bandwagon effect and the differentiation of candidates.\\
~~\\
Our analysis is connected to the literature concerned with peer effects (see \cite{bramoulle2020peer} for a recent survey). Our approach however differs from the study of peer effects on networks: we do not assume any structure between consumers\footnote{However, in our model, all consumers buying the same product are on the same side of the indifferent consumer, which is an assortative effect.} and simplify the diffusion by supposing that the continuum of consumers is simply influenced by the mass of consumers selecting the same product. This simplification allows us to construct a first-stage game, where the characteristics of the products are chosen by firms who anticipate the bandwagon effect. In our setting, only firms, and not consumers, have strategic interactions.  

\section{The model}
Two firms simultaneously select their location $x_1,x_2$ on the characteristic space $[0,1]$. We assume without loss of generality that $x_1\leq x_2$. After locations are announced, consumers buy one unit of the good that maximizes their utility. We assume that consumer $v$'s utility to select good $i$ is given by:
$$u_v^i=\theta+a s_i - |v-x_i|$$
where $s_i$ is the quantity of consumers that select product $i$ (also referred to as the market share of firm $i$), $a \in (0,1)$ is the relative magnitude of the positive externality of the market share, and $\theta$ is the intrinsic utility of buying the good (net of its fixed price). Because we suppose that the demand is inelastic, we assume that $\theta\geq 1$, which is sufficiently large to explain why consumers always prefer to buy.\\
Note that computing $s_i$ requires to compute $u^i_v$, whose expression involves in turn $s_i$. Therefore, to precise what is the market reaction after locations $x_1,x_2$ are chosen, we define a market equilibrium as follows:

\begin{definition}[Market equilibrium]
After firms have selected locations $(x_1,x_2)$, we say that $(s_1,s_2)$ is a market equilibrium if $s_2=1-s_1$, $0\leq s_1\leq1$ and for all $v\in[0,s_1)$, 
$$u_v^1 \geq u_v^2,$$
and for all $v\in[s_1,1)$, 
$$u_v^1 \leq u_v^2$$
\end{definition}


In the case where an interval of consumers $[v_1,v_2]$ is indifferent between buying from $x_1$ or $x_2$, with $x_1<x_2$, we suppose for simplicity and without loss of generality that the left interval $[v_1,\tilde{v}]$ shops to $x_1$ and the right interval $[\tilde{v},v_2]$ shops to $x_2$, for a certain $\tilde{v}$.

A key result for our paper is that a multiplicity of market equilibria can exist, as the following example illustrates.

\begin{example}
Assume that $a=\frac12$ and $(x_1,x_2)=(\frac13,\frac23)$.
Then $(s_1,s_2)\in\{(0,1),(\frac16,\frac56),(\frac12,\frac12),(\frac56,\frac16),(1,0)\}$ are all market equilibria. To understand why $(\frac16,\frac56)$ is one of the possible market equilibria, we show that consumer $v=\frac{1}{6}$ is indifferent between buying from firm $1$ or $2$ when market shares are indeed equal to $(\frac16,\frac56)$: buying from firm $1$ provides him a utility equal to $\theta+\frac{1}{2}\cdot\frac{1}{6}-|\frac13-\frac16|=\theta-\frac{1}{12}$ while buying from firm $2$ provides him a utility equal to $\theta+\frac{1}{2}\cdot\frac{5}{6}-|\frac{1}{6}-\frac23|=\theta-\frac{1}{12}$. The same argument holds for all consumers in $v\in [0,\frac{1}{3}]$ while $u_v^1>u_v^2$ for all consumers in $v\in (\frac{1}{3},1]$.

The fact that the indifferent consumer can locate either in $(x_1,x_2)$, in $(0,x_1]$, in $[x_2,1)$ or at the boundaries $0$ or $1$ provides up to $5$ possible market equilibria.

\end{example}

We define firms' payoffs as their market shares $\pi_1=s_1$ and $\pi_2=s_2$. In the first motivating example provided in the introduction, this hypothesis is made without loss of generality as long as the marginal production cost never exceeds the fixed price, so it is always beneficial to sell more. In the second motivating example, this hypothesis supposes that we analyze a proportional representative election, or that there is imperfect information (as discussed in the introduction). In order to decide whether a quadruple $(x_1,x_2,s_1,s_2)$ is an equilibrium, we have to cope with the multiplicity of market equilibria.\\
~~\\
When firms select $(x_1,x_2)$, they cannot anticipate their payoffs as the consumers' reaction can lead to different market shares. We take advantage of this multiplicity of market equilibria to study different firms' behaviors regarding the uncertainty. Formally, the behavior of a firm is a function that maps any deviation $(x'_1,x'_2)\neq (x_1,x_2)$ to a deviation vector payoff $(s_1(x'_1,x'_2),s_2(x'_1,x'_2))$, which is computed by means of the market equilibria. We focus our attention on the three possible behaviors:

\begin{definition}\label{def:behav}[Firms' behaviors]~~\\
If firms assume that the market equilibrium that minimizes their own payoff will occur after locations are chosen (worst-case scenario) we say that they are are \textbf{pessimistic}.\\ 
If firms assume that the market equilibrium that maximizes their own payoff will occur after locations are chosen (best-case scenario) we say that they are are \textbf{optimistic}.\\
Finally, if firms suppose that every market equilibrium has the same probability to occur after locations are chosen (average scenario) we say that they are are \textbf{neutral}.
\end{definition}

When firms' behaviors are fixed, we say that a quadruple $(x_1,x_2,s_1,s_2)$ is a \textit{Nash equilibrium (NE)} if $(s_1,s_2)$ is a market equilibrium for $(x_1,x_2)$, and $s_1\geq s_1(x'_1,x_2)$ for all $x'_1\neq x_1$ and $s_2\geq s_2(x_1,x'_2)$ for all $x'_2\neq x_2$.


\section{Characterization of Nash Equilibria}

We first study market equilibria by assuming that the pair of locations $(x_1,x_2)$ are fixed.

\subsection{Characterization of market equilibria}

We still assume, without loss of generality, that $x_1 \leq x_2$.
\begin{proposition} \label{pro:vote}
If $x_2-x_1>a$, then there is a unique market equilibrium:
\begin{enumerate}
\item[(i)] $s_1=\frac{x_1+x_2-a}{2(1-a)}$. In this case $[0, s_1)$ strictly prefer firm $1$ and $(s_1,1]$ strictly prefer firm $2$.
\end{enumerate}
If $x_2-x_1\leq a$, then there exist up to $5$ market equilibria:
\begin{enumerate}
\item[(i)] $s_1=0,s_2=1$. In this case, all consumers prefer firm $2$. 
\item[(ii)] $s_1=\frac{1}{2}-\frac{x_2-x_1}{2a}$. In this case, $[0,x_1)$ are indifferent and $(x_1,1]$ strictly prefer firm $2$. This firm exists if and only if $$s_1 \in [0,x_1) \Leftrightarrow  a\leq x_2-(1-2a)x_1.$$
\item[(iii)] $s_1=\frac{x_1+x_2-a}{2(1-a)}$. In this case $[0, s_1)$ strictly prefer firm $1$ and $(s_1,1]$ strictly prefer firm $2$. This firm exists if and only if $$s_1 \in [x_1,x_2] \Leftrightarrow x_1-(1-2a)x_2 \leq a \leq x_2-(1-2a)x_1.$$
\item[(iv)] $s_1=\frac{1}{2}+\frac{x_2-x_1}{2a}$. In this case $[0,x_2]$ strictly prefer firm $1$ and $(x_2,1]$ are indifferent. This firm exists if and only if $$s_1 \in (x_2,1] \Leftrightarrow x_1-(1-2a)x_2 \leq a.$$
\item[(v)] $s_1=1,s_2=0$. In this case, all consumers prefer firm $1$.
\end{enumerate}
\end{proposition} 

The economic interpretation of the two distinct scenarios is the following: in case where $x_2-x_1>a$, the two products exhibit distant characteristics. In this setting, comparing the products' characteristics is very significant for the consumers' choice, and the influence of a product's popularity is diminished. We find that there the indifferent consumer lies in the interval between $x_1$ and $x_2$ as in the case where the popularity effect does not exist. More precisely, we find that the indifferent consumer locates at $\frac{x_1+x_2-a}{2(1-a)}$, which converges to $\frac{x_1+x_2}{2}$ when $a \rightarrow 0$, as commonly assumed in location games without popularity effects ($a=0$).\\
On the other hand, when firms opt for proximate locations, the popularity of a product assumes a critical role. With high enough popularity, firm $2$ attracts a consumer $v$ located on the left of firm $1$. The amplifying effect of the popularity generates outcomes that are less stable and we obtain a multiplicity of equilibria.

\begin{proof}
First assume that $x_2-x_1 > a$. We solve the equation $u_v^1=u_v^2$ to find the indifferent consumer $v$. We start by proving that the case where $v < x_1$ is not possible:
$u_v^1 \geq u_v^2$ is equivalent to $\theta+av-x_1-v=\theta+a(1-v)-x_2-v$ which simplifies to $2av-a+x_2-x_1=0$. Because we assumed $x_2-x_1 > a$, we find that $v<0$. In other words, we should have $u_0^1<u_0^2$, which contradicts our hypothesis. A symmetric argument implies that the case $v> x_2$ is not neither possible. We conclude with the case where $v \in (x_1,x_2)$ where solving $u_v^1 \geq u_v^2$ leads to $v=\frac{x_1+x_2-a}{2(1-a)}$, which is always a valid solution as it belongs to $(x_1,x_2)$ as soon as $x_2-x_1>a$.\\
Assume now that $x_2-x_1 \leq a$. We repeat the same exercise, assuming respectively that $v \in (0,x_1)$, $v \in (x_2,1)$, $v \in [x_1,x_2]$ and $v \in \{0,1\}$. However, we find no contradiction in this case. In each case, the constraint that the indifferent consumer belongs to the appropriate interval provides necessary and sufficient conditions.  
\end{proof}

\begin{corollary}
There exists $1$, $3$ or $5$ market equilibria.\\
More precisely, there exists a unique market equilibrium when $x_2-x_1 > a$ and there exists either $3$ or $5$ market equilibria when $x_2-x_1\leq a$.
\end{corollary}

\begin{proof}
By considering the necessary and sufficient existence conditions provided in Proposition \ref{pro:vote} in the case where $x_2-x_1\leq a$, we find that:\\
- equilibria of type $(i)$ and $(v)$ always exist.\\ 
- if $(ii)$ does not exist, then $(iv)$ exists, and if $(iv)$ does not exist, then $(ii)$ exists.\\ 
- $(ii)$ and $(iv)$ both exist if and only if $(iii)$ exists.
\end{proof}

We are now ready to study the choice of locations by firms. As described in definition \ref{def:behav}, we study successively the case where firms are optimistic, neutral and pessimistic.

\subsection{Optimistic firms}
Suppose first that firms are optimistic. They assume that the market equilibrium that will emerge after a deviation of their location from the equilibrium is the one that maximizes their own payoff.  

\begin{proposition}
When firms are optimistic, there is no NE.    
\end{proposition}
\begin{proof}
Indeed, suppose that an NE exists and denote $(x_1,x_2,s_1,s_2)$ the equilibrium locations and markets shares. If firm $1$ deviates to $x_1' \in [x_2-a,x_2+a]$, then $(s_1,s_2)=(1,0)$ is a possible market equilibrium, therefore firm $1$'s deviation payoff is $1$. A symmetric argument applies to firm $2$, so one of the two deviations is strictly profitable as $s_1+s_2=1$.
\end{proof}

The economic interpretation is straightforward: when a firm deviates to locate close to its opponent, the multiplicity of market equilibria includes the extreme case where one firm attracts all consumers. A optimistic firm assumes that its market share will be $1$ (and ignores the case where it loses all the market share). There is therefore no stable situation.

\subsection{Neutral firms}
Suppose now that firms are neutral. They assume that every market equilibrium can emerge with the same probability after a deviation of their location from the equilibrium, and compute their excepted payoff. 

\begin{proposition}\label{pro:NEneutral}
When firms are neutral:\\
There is no NE for $a>\frac12$.\\
There exists a unique NE for $a \leq \frac12$, where $(x_1^*,x_2^*)=(\frac12,\frac12)$ and $s_1=s_2=\frac12$.
\end{proposition}

\begin{proof}
Suppose that an NE exists and denote $(x_1,x_2,s_1,s_2)$ the equilibrium locations and markets shares. By deviating to $1-x_2$, firm $1$'s deviation payoff is $\frac12$, either because it is the unique market equilibrium (when $x_2-x_1>a)$ or because market equilibrium $(iii)$, and thus all five type of market equilibria exist (when $x_2-x_1 \leq a$) and the payoffs in market equilibria of type $(i)$ and $(v)$ sum to $1$, as well as market equilibria of type $(ii)$ and $(iv)$ sum to 1. A symmetric argument applies to firm $2$. So in an NE, we have that $s_1=s_2=\frac12$.

Now, assume by contradiction that $x_2>\frac12$. 
By choosing $x_1=1-x_2+\epsilon$, with $\epsilon>0$ small enough, firm 1 obtains a market share strictly larger than $\frac12$: in the case where $x_2-x_1>a$, we have that $\frac{1+\epsilon-a}{2(1-a)}$ increases with $\epsilon$, in the case where $x_2-x_1 \leq a$ we have that market equilibrium $(iii)$ exists and the sum of the five market shares also increases with $\epsilon$. Hence there is no equilibrium in which $x_2>\frac12$. A symmetric argument holds for firm $1$ and we proved that in an NE, $x_1=x_2=\frac12$.

It remains to show that $(\frac12,\frac12,\frac12,\frac12)$ is an NE for $a\leq \frac12$ and not an NE for $a>\frac12$. If $a\leq \frac12$ and $x_2=\frac12$, then market equilibrium $(iii)$ exists for all $x_1\in[0,\frac12]$. But any deviation to $x_1 \in [0,\frac12]$ yields a lower average market share for firm $1$. A symmetric argument holds for a deviation to $[\frac12,1]$

If $a>\frac12$, $x_1=0$ and $x_2=\frac12$, then market equilibrium $(iii)$ does not exist, but market equilibrium $(iv)$ exists. This is a beneficial deviation for firm 1.
\end{proof}

In the case where firms are neutral, they compute their expected payoff among possible market equilibria. A key result is that payoffs in market equilibria of type $(i)$ and $(v)$ sum to $1$, as well as market equilibria of type $(ii)$ and $(iv)$. Therefore, only the firm of type $(iii)$ varies with location, and it has the crucial property to be increasing with $x_1$ and $x_2$. Therefore, firms want to locate as close as possible to each other.

As an illustration, consider the case where $a \rightarrow 0$. We then have that all five market equilibria exist and, similarly to the case $a=0$ where the principle of minimal differentiation holds, firms minimize their differentiation by selecting the same location at equilibrium. Our proposition proves that this result is robust to the introduction of a popularity effect as long as its magnitude, captured by the parameter $a$, is smaller than $\frac12$. 

On the other hand, when $a$ is larger than $\frac12$, our finding mitigates the principle of minimum differentiation. Indeed, we find that, against $x_2=\frac12$, it is better for firm $1$ to deviate from $x_1=\frac12$ to $x_1'<\frac12$. The intuition is that, if firm $1$ deviates to $x_1'=0$ for example, the possible market equilibria are $s_1=0$, $s_1=\frac12 + \frac{1}{4a}$ or $s_1=1$, whose average is larger than $\frac12$. By selecting what seems to be an unfavorable location (the midpoint between $x_1'$ and $x_2$ is now strictly smaller than $\frac12$), firm $1$ reduces the set of possible market equilibria to $3$ configurations, two of them being in its favor.

\subsection{Pessimistic firms}

We now turn to the case where firms are pessimistic. We find a more complex equilibrium structure.

\begin{proposition}\label{prop:NE}
    A quadruple $(x_1,x_2,s_1,1-s_1)$  is an NE if and only if $(s_1,1-s_1)$ is a market equilibrium for $(x_1,x_2)$ and
\begin{enumerate}
    \item[(i)] $s_1\in\left[\frac{1-x_2-a}{1-a},\frac{x_1}{1-a}\right]$ if $x_1,x_2\leq \frac12$.
    \item[(ii)] $s_1\in[\frac{x_2-a}{1-a},\frac{x_1}{1-a}]$ if $x_1\leq \frac12\leq x_2$.
    \item[(iii)] $s_1\in[\frac{x_2-a}{1-a},\frac{1-x_1}{1-a}]$ if $x_1,x_2\geq \frac12$.
\end{enumerate}
\end{proposition}

\begin{proof}
We detail the proof for case $(i)$. The other cases follow similarly. 
A quadruple $(x_1,x_2,s_1,1-s_1)$ is an NE if and only if $(s_1,1-s_1)$ is a market equilibrium for $(x_1,x_2)$ and no firm has a beneficial deviation. We conclude using the following claim, whose proof is postponed to the appendix:\\
Claim 1: Let $x_j\leq \frac12$, and $i\neq j$ be the other firm. The three following statements are equivalent: (1) Candidate $i$ has no beneficial deviation, (2) If $x_j+a \leq 1$, then candidate $i$ does not benefit from deviating to $x_j+a$, (3) $s_i \geq 1-\frac{x_j}{1-a}$.\\
We conclude from Claim 1 that firm 1 has no beneficial deviation if and only if  $s_1\geq 1-\frac{x_2}{1-a}=\frac{1-x_2-a}{1-a}$ and that firm 2 has no beneficial deviation if and only if  $s_2=1-s_1\geq 1-\frac{x_1}{1-a} \Leftrightarrow s_1\leq \frac{x_1}{1-a}$.
\end{proof}

To characterize explicitly the set of NE in the case of pessimistic firms, we need to find a quadruplet $(x_1,x_2,s_1,s_2)$ that satisfies conditions from Proposition \ref{pro:vote} and Proposition \ref{prop:NE} simultaneously. Proposition \ref{pro:ex_sym} and Proposition\ref{pro:ex_frac12} illustrate this characterization respectively in the case where locations are symmetric and where $a=\frac12$.\\

First note that the case where firms are pessimistic differs significantly from the case where firms are neutral. In the later case, we proved that the unique equilibrium is completely symmetric $x_1=x_2$ and $s_1=s_2=\frac12$. In the case where firms are pessimistic, we find that the minimal differentiation principle does not hold in general, and that there are equilibria with $x_1 \neq x_2$. However, the distance between the two locations is bounded above by the parameter $a$, therefore, we still have that $x_2 \rightarrow x_1$ when $a \rightarrow 0$ at equilibrium. Similarly, the asymmetry between the market shares is bounded above by a quantity that converges to $0$ when $a \rightarrow 0$.

\begin{proposition}
If $(x_1,x_2,s_1,s_2)$ is an NE, then
$$|x_2-x_1| \leq a,$$
and 
$$|s_2-s_1|\leq \frac{a}{1-a}.$$ 
\end{proposition}
\begin{proof}
\begin{itemize}
\item We first prove that $|x_2-x_1| \leq a$.  Suppose that $x_2-x_1>a$. Then by Proposition \ref{pro:vote}, there is a unique market equilibrium with $s_1=\frac{x_1+x_2-a}{2(1-a)}$. Since this $s_1$ increases in $x_1$, it is a profitable deviation for candidate 1 to move to the right, as long as the constraint $x_2-x_1 >a$ still holds.
   
\item We now prove that $s_1\geq \frac{1-2a}{2(1-a)}$ at equilibrium. Because a symmetric inequality holds for $s_2$ we obtain $|s_2-s_1|\leq \frac{a}{1-a}$.\\
Observe that the result is trivially true for $a\geq\frac12$. So we can assume that $a<\frac12$. Given that there are three different cases in Proposition \ref{prop:NE} and five different market equilibria (Proposition \ref{pro:vote}), we have to consider 15 cases. We illustrate the two equilibria that yield the lowest market shares. The other 13 cases can be analyzed similarly.

    Case 1. Assume that $x_1\leq x_2\leq \frac12$ and consider market equilibrium $(ii)$. Then we want to minimize the market share of firm 1 while guaranteeing that $(x_1,x_2)$ and $(s_1,1-s_1)$ defined by market equilibrium $(ii)$ is an NE, that is, we want to solve
    $$\min\left\{s_1\mid 0\leq \frac{1}{2}-\frac{x_2-x_1}{2a}\leq x_1\leq x_2\leq\frac12, x_2-x_1\leq a, \frac{1-x_2-a}{1-a}\leq  \frac{1}{2}-\frac{x_2-x_1}{2a}\leq\frac{x_1}{1-a}\right\}.$$
    Solving the above optimization problem yields $s_1=\frac{1-2a}{2(1-a)}$ for $x_1=\frac{(1+a)\cdot(1-2a)}{2(1-a)}$ and $x_2=\frac12$.

    Case 2. Assume that $x_1\leq x_2\leq \frac12$ and consider market equilibrium $(iii)$. Then we want to minimize the market share of firm 1 while guaranteeing that $(x_1,x_2)$ and $(s_1,1-s_1)$ defined by market equilibrium $(iii)$ is an NE, that is, we want to solve
    $$\min\left\{s_1\mid 0\leq x_1\leq \frac{x_1+x_2-a}{2(1-a)}\leq x_2\leq\frac12, x_2-x_1\leq a, \frac{1-x_2-a}{1-a}\leq  \frac{x_1+x_2-a}{2(1-a)}\leq\frac{x_1}{1-a}\right\}.$$
    Solving the above optimization problem yields $s_1=\frac{1-2a}{2(1-a)}$ for $x_1=\frac{1-2a}{2}$ and $x_2=\frac12$.\end{itemize}\end{proof}

\subsubsection{First illustrative example: the case of symmetric equilibria}

We now illustrate Proposition \ref{prop:NE} to explicitly characterize the set of NE in which $x_2=1-x_1$. 

\begin{proposition}\label{pro:ex_sym}
A quadruple $(x_1,1-x_1,s_1,1-s_1)$ is an NE if and only if 
\begin{enumerate}
\item $\frac{1-a}{2}\leq x_1\leq \frac{1-a^2}{2}$, and $s_1=\frac12$.
\item $\frac{1-a^2}{2}\leq x_1 \leq\min\{1-a,\frac{1}{2}\}$, and $s_1=\frac{1}{2}$, $s_1=\frac{1}{2}-\frac{1-2x_1}{2a}$ or $s_1=\frac{1}{2}+\frac{1-2x_1}{2a}$. 
\item $1-a\leq x_1 \leq\frac{1}{2}$, and $s_1=\frac{1}{2}$, $s_1=\frac{1}{2}-\frac{1-2x_1}{2a}$, $s_1=\frac{1}{2}+\frac{1-2x_1}{2a}$, $s_1=0$ or $s_1=1$. 
\end{enumerate} 
\end{proposition}

Figure \ref{fig:sym} below illustrates the $3$ different regions. In region $1$, equilibrium locations are relatively distant from each other but market shares are equal. In region $2$, the distance between locations is smaller and market shares can be asymmetric, but both firms have a strictly positive clientele. In region $3$, the distance between locations is arbitrarily small, a situation where one firm covers the entire market is possible.

\begin{figure}[h]
\centering
\includegraphics[scale=0.7]{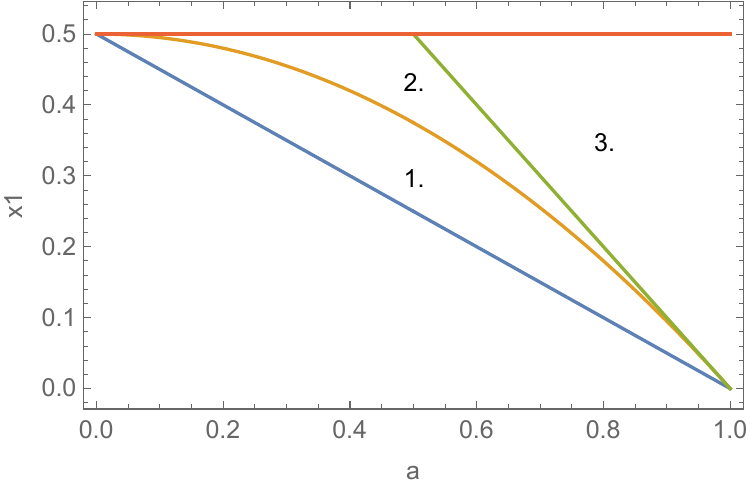}
\caption{\label{fig:sym} Possible symmetric equilibria $(x_1,1-x_1)$ with respect to $a$.}
\end{figure}

\begin{proof}
 We check the conditions of Proposition \ref{prop:NE} for the five different market equilibria given in Proposition \ref{pro:vote}. 

\begin{enumerate}
\item[(i)] Observe that $0\in[\frac{1-x_1-a}{1-a},\frac{x_1}{1-a}]$ is equivalent to $x_1\geq 1-a$. By Theorem \ref{prop:NE}, $s_1=0$ is an NE if and only if $s_1=0$ is a market equilibrium and $x_1 \geq 1-a$. Remains to prove that $s_1=0$ is a market equilibrium. Because $1-x_1-x_1\leq 2a-1\leq a$, by Proposition \ref{pro:vote} $s_1=0$ is a market equilibrium.
\item[(ii)] Observe that $\frac{1}{2}-\frac{1-2x_1}{2a}\in[\frac{1-x_1-a}{1-a},\frac{x_1}{1-a}]$ is equivalent to $x_1 \geq \frac{1-a^2}{2}$. By Theorem \ref{prop:NE}, $s_1=\frac{1}{2}-\frac{1-2x_1}{2a}$ is an NE if and only if $s_1=\frac{1}{2}-\frac{1-2x_1}{2a}$ is a market equilibrium and $x_1 \geq \frac{1-a^2}{2}$. Remains to prove that $s_1=\frac{1}{2}-\frac{1-2x_1}{2a}$ is a market equilibrium. Because $1-x_1-x_1\leq 1-2\cdot\frac{1-a^2}{2}=a^2<a$ and $1-x_1-(1-2a)x_1=1-2(1-a)x_1\geq a$, by Proposition \ref{pro:vote} $s_1=\frac{1}{2}-\frac{1-2x_1}{2a}$ is a market equilibrium. 

\item[(iii)] Observe that $\frac12\in[\frac{1-x_1-a}{1-a},\frac{x_1}{1-a}]$ is equivalent to $x_1 \geq \frac{1-a}{2}$. By Theorem \ref{prop:NE}, $s_1=\frac12$ is an NE if and only if $s_1=\frac{1}{2}$ is a market equilibrium and $x_1 \geq \frac{1-a}{2}$. Remains to prove that $s_1=\frac{1}{2}$ is a market equilibrium. Because $1-x_1-x_1\leq 1-2\cdot\frac{1-a}{2}=a$ and $x_1-(1-2a)(1-x_1)=2a-1+2x_1(1-a)\leq a\leq 1-2(1-a)x_1=1-x_1-(1-2a)x_1$, by Proposition \ref{pro:vote} $s_1=\frac{1}{2}$ is a market equilibrium.

\item[(iv)] Observe that $\frac{1}{2}+\frac{1-2x_1}{2a}\in[\frac{1-x_1-a}{1-a},\frac{x_1}{1-a}]$ is equivalent to $x_1 \geq \frac{1-a^2}{2}$. By Theorem \ref{prop:NE}, $s_1=\frac{1}{2}+\frac{1-2x_1}{2a}$ is an NE if and only if $s_1=\frac{1}{2}+\frac{1-2x_1}{2a}$ is a market equilibrium and $x_1 \geq \frac{1-a^2}{2}$. Remains to prove that $s_1=\frac{1}{2}+\frac{1-2x_1}{2a}$ is a market equilibrium. Because $1-x_1-x_1\leq 1-2\cdot\frac{1-a^2}{2}=a^2<a$ and $x_1-(1-2a)(1-x_1)=2a-1+2x_1(1-a)\leq a$, by Proposition \ref{pro:vote} $s_1=\frac{1}{2}+\frac{1-2x_1}{2a}$ is a market equilibrium. 

\item[(v)] Observe that $1\in[\frac{1-x_1-a}{1-a},\frac{x_1}{1-a}]$ is equivalent to $x_1\geq 1-a$. By Theorem \ref{prop:NE}, $s_1=1$ is an NE if and only if $s_1=1$ is a market equilibrium and $x_1 \geq 1-a$. Remains to prove that $s_1=1$ is a market equilibrium. Because $1-x_1-x_1\leq 2a-1\leq a$, by Proposition \ref{pro:vote} $s_1=1$ is a market equilibrium.
\end{enumerate}
Summarizing the above NE yields the result.
\end{proof}

\subsubsection{Second illustrative example: the case where $a=\frac12$.}

We illustrate Proposition \ref{prop:NE} in the case where $a=\frac12$. 

\begin{proposition}\label{pro:ex_frac12}
If $a=\frac12$, the pair $(x_1,x_2)$ is an NE for the corresponding market equilibrium type if
\begin{enumerate}
    \item[(i)] $0\leq x_1\leq \frac12$ and $x_2=\frac12$.
    \item[(ii)] $0\leq x_1\leq \frac12$ and $\frac12\leq x_2\leq \frac{3+2x_1}{6}$, and $\frac12\leq x_1\leq \frac34$ and $x_1\leq x_2\leq \frac{3+2x_1}{6}$.
    \item[(iii)] $0\leq x_1\leq \frac12$ and $\frac12\leq x_2\leq \frac{1+2x_1}{2}$.
    \item[(iv)] $\frac14\leq x_1\leq \frac12$ and $x_1\leq x_2\leq \frac{-1+6x_1}{2}$.
    \item[(v)] $x_1=\frac12$ and $\frac12\leq x_2\leq 1$.
\end{enumerate}
\end{proposition}

\begin{proof}
These NE are found following a case by case analysis that is provided in the appendix (Section \ref{se:appendix}).  
\end{proof}

Figure \ref{fig:a12} illustrates the above set of NE. Observe that the set of NE is non-convex even for fixed $x_1$.

\begin{figure}[h]
\centering
\begin{tikzpicture}[scale=0.5]
\begin{scope}[blend mode=multiply]
\draw[blue!50,fill=blue!50] (0,5)--(5,5)--(5,10);
\draw[red!50,fill=red!50] (0,5)--(5,5)--(7.5,7.5);
\draw[green!50,fill=green!50] (2.5,2.5)--(5,5)--(5,10);
\end{scope}
\draw[thick] (0,5) to (5,5);
\draw[thick] (5,5) to (5,10);
\draw[-stealth] (0,0) node[below left]{0} to (10.5,0) node[right]{$x_1$};
\draw[-stealth] (0,0) to (0,10.5) node[above]{$x_2$};
\draw[dashed] (0,0) to (10,10);
\draw[|-|] (2.5,0) node[below] {$\frac14$} to  (5,0) node[below] {$\frac12$};
\draw[|-|] (7.5,0) node[below] {$\frac34$} to (10,0) node[below] {1};
\draw[|-|] (0,2.5) node[left] {$\frac14$} to (0,5) node[left] {$\frac12$};
\draw[|-|] (0,7.5) node[left] {$\frac34$} to (0,10) node[left] {1};
\node at (2,4.5) {$(i)$};
\node at (5.75,6.5) {$(ii)$};
\node at (3,7) {$(iii)$};
\node at (3.75,4.5) {$(iv)$};
\node at (5.75,8) {$(v)$};
\end{tikzpicture}\caption{Set of equilibria for $a=\frac12$.}\label{fig:a12}
\end{figure}
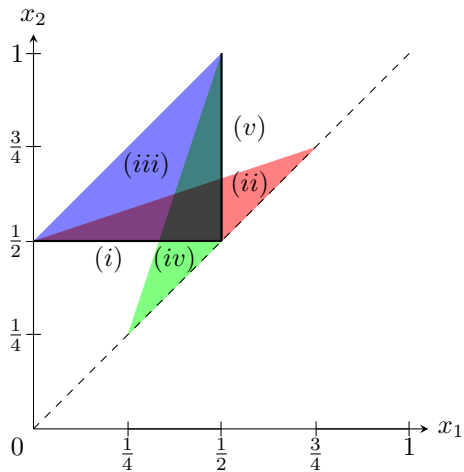

\section{Welfare analysis}

In this section, we investigate the welfare impact of the popularity effect. Because the total payoff of firms is constantly equal to $1$, we exclude them from the welfare analysis and restrict our attention to the consumers' side.\\
~~\\
The comparative static on $a$ of the consumers' welfare is not trivial. Indeed, for fixed locations $(x_1,x_2)$, consumers' utility is increasing with $a$, but this parameter also plays a role when firms select their locations $(x_1,x_2)$ at equilibrium. Parameter $a$ has an impact on the distance between firms and thus the distance between consumers and firms. Both in industrial organisation and in political economy, many papers are concerned with the inefficiency of equilibrium: firms usually do not differentiate enough in comparison with the social optimum. Popularity modifies this perspective: we find that when $a$ is large enough, consumers might not care about the differentiation of firms and could prefer to buy from a popular firm even if the product differs significantly from their ideal. For example, we find that when $a \geq \frac14$, the social optimum is such that all consumers buy from a central firm at $\frac12$ and do not care about the location of the second firm.\\
~~\\
Because our game typically display a multiplicity of equilibrium, we measure the price of anarchy and the price of stability (respectively the ratio between the social cost of the worst / best NE and the social optimum). 

Our main findings are Proposition \ref{pro:poa} and Proposition \ref{pro:pos}, these ratios are non-monotonic with respect to the magnitude of the popularity effect.

\subsection{Social welfare}

The welfare function measures the consumers' surplus, that is the integral of their utilities, for fixed locations $(x_1,x_2)$, market shares $(s_1,1-s_1)$ and for a fixed parameter $\theta \geq 1$. Formally,

\begin{align*}
    W_a(x_1,x_2,s_1,1-s_1,\theta)&=\int_0^{s_1} \left(\theta+ a s_1 - |t-x_1|\right)dt+\int_{s_1}^1 \left(\theta+ a (1-s_1) - |t-x_2|\right)dt\\
    &=\theta+ a s_1^2+a(1-s_1)^2 - \int_0^{s_1}|t-x_1|dt-\int_{s_1}^1 |t-x_2|dt.
\end{align*}

And we define $PoA(a,\theta)$ and $PoS(a,\theta)$ the price of anarchy and the price of stability with parameter $a$ and $\theta$ as follows:

$$PoA(a,\theta)= \frac{\displaystyle\max_{x_1,x_2,s_1} W_a(x_1,x_2,s_1,1-s_1,\theta)}{\displaystyle\min_{(x_1^*,x_2^*,s_1^*) \in NE}  W_a(x_1^*,x_2^*,s_1^*,1-s_1^*,\theta)}$$

$$PoS(a,\theta)= \frac{\displaystyle\max_{x_1,x_2,s_1} W_a(x_1,x_2,s_1,1-s_1,\theta)}{\displaystyle\max_{(x_1^*,x_2^*,s_1^*) \in NE}  W_a(x_1^*,x_2^*,s_1^*,1-s_1^*,\theta)}$$

Our objective is to compare the welfare at equilibrium with the social optimum. The following proposition provides a simple expression for the welfare of the social optimum.
\begin{proposition}\label{pro:opt}{Social optimum}~~\\
If $a< \frac14$, then $x_1=\frac14$, $x_2=\frac34$ and $s_1=\frac12$ is the social optimum. The optimal welfare is then $$W_a\left(\frac14,\frac34,\frac12,\frac12,\theta\right)= \theta - \frac18 + \frac{a}{2}$$
If $a\geq \frac14$, then $x_2=\frac12$ and $s_1=0$, or $x_1=\frac12$ and $s_1=1$ are social optima. The optimal welfare is then $$W_a\left(0,\frac12,0,1,\theta\right)= \theta - \frac14 + a$$
\end{proposition}

Interestingly, we find a cut-off value $a=\frac14$: below this threshold, the popularity plays no role. At the social optimum, firms locate as a planner minimizing only the average distance between firms and consumers would select. Above this threshold, the popularity matters first and only one firm covers the demand. This firms maximizes the consumers' surplus by locating at $\frac12$, and the second firm has no impact.
\begin{proof}
    We solve the following two optimization problems:
    \begin{align}\label{eq1}
     \max_{x_1,x_2,s_1} \quad&a s_1^2+a(1-s_1)^2 - \int_0^{x_1}(x_1-t)dt-\int_{x_1}^{s_1}(t-x_1)dt-\int_{s_1}^{x_2}(x_2-t)dt-\int_{x_2}^1(t-x_2)dt\\
    s.t. \quad & 0\leq x_1\leq s_1\leq x_2\leq 1.\nonumber
    \end{align}
    and 
    \begin{align}\label{eq2}
    \max_{x_1,x_2,s_1} \quad&a s_1^2+a(1-s_1)^2 - \int_0^{x_1}(x_1-t)dt-\int_{x_1}^{s_1}(t-x_1)dt-\int_{s_1}^1(t-x_2)dt\\
    s.t. \quad & 0\leq x_1\leq x_2\leq s_1\leq 1.\nonumber
    \end{align}
    If $a< \frac14$, then \eqref{eq1} has a larger objective value with optimal values $x_1^*=\frac14$, $x_2^*=\frac34$ and $s_1^*=\frac12$. If $a=\frac14$, then \eqref{eq1} and \eqref{eq2} have the same objective value. If $a\geq \frac14$, then \eqref{eq2} has a larger objective value with optimal values $x_2^*=\frac12$ and $s_1^*=0$, or $x_1^*=\frac12$ and $s_1^*=1$.
\end{proof}

While we cannot compute the PoA and the PoS for optimistic firms (because there exist no NE), we study seperately the case where firms are neutral or pessimistic.

\subsection{Neutral firms}

In the case where firms are neutral, the game admits a unique NE for $a \leq \frac12$ and no NE for $a>\frac12$. Therefore, $PoA(\theta,a)=PoS(\theta,a)$ for any $a\leq \frac12$. More precisely:

\begin{proposition}
$$PoA(a,\theta)=PoS(a,\theta)=\begin{cases}
\displaystyle\frac{\theta-(\frac18-\frac{a}{2})}{\theta-(\frac14-\frac{a}{2})}& \text{ if }a\leq \frac14,\\
~~\\
\displaystyle\frac{\theta-(\frac14-a)}{\theta-(\frac14-\frac{a}{2})}& \text{ if }\frac14<a\leq \frac12.
\end{cases}$$
\end{proposition}

\begin{proof}
The proposition follows directly from Proposition \ref{pro:opt} and Proposition \ref{pro:NEneutral}.
\end{proof}

\begin{figure}[h]
\centering
\includegraphics[scale=0.7]{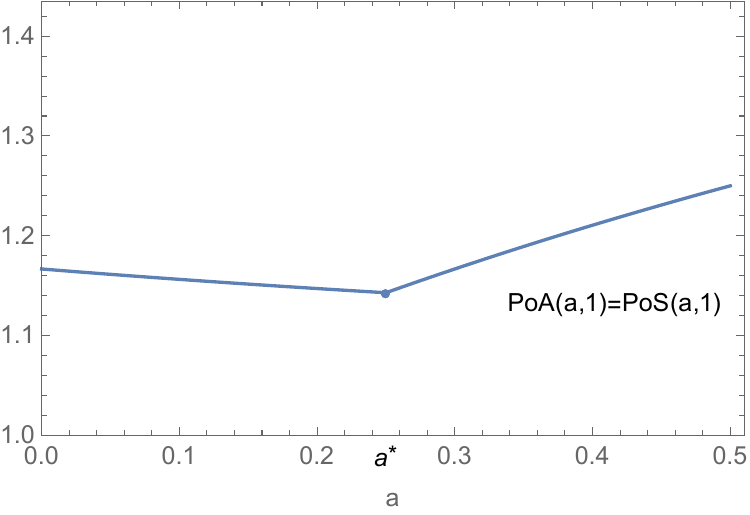}
\caption{\label{fig:poa} Price of anarchy and price of stability for neutral firms as a function of $a$, for $\theta=1$.}
\end{figure}

We find that the PoA and the PoS are decreasing with $a \in (0,\frac14]$ and increasing with $a \in [\frac14,\frac12]$, reaching their minimum for $a=\frac14$. The minimum is therefore equal to $\frac{\theta}{\theta-\frac18}$. For $\theta=1$, we find a minimal loss of efficiency equal to $\frac{8}{7}\simeq 1.14$.

\subsection{Pessimistic firms}

\subsubsection{Price of anarchy}

\begin{proposition} \label{pro:ana}
The NE that minimizes the social welfare is $x_1=x_2=\frac{1-a}{2}$ and $s_1=\frac12$ with a welfare of $\theta-\frac{1+a^2}{4}+\frac{a}{2}=\theta - \frac{(1-a)^2}{4}$.
\end{proposition}

\begin{proof}
    Given that there are three different cases in Proposition \ref{prop:NE} and five different market equilibria (Proposition \ref{pro:vote}), we have to consider 15 cases. We illustrate the equilibrium that yield the lowest welfare. The other 14 cases can be analyzed similarly.

    Assume that $x_1\leq x_2\leq \frac12$ and consider market equilibrium $(iv)$. Then we want to minimize the welfare while guaranteeing that $(x_1,x_2)$ and $(s_1,1-s_1)$ defined by market equilibrium $(iv)$ is an NE, that is, we want to solve
    \begin{align*}
    \displaystyle\min_{s_1,x_1,x_2}\quad&\theta+ a s_1^2+a(1-s_1)^2 - \int_0^{x_1}(x_1-t)dt-\int_{x_1}^{s_1}(t-x_1)dt-\int_{s_1}^1 (t-x_2)dt  \\
    &0\leq x_1\leq x_2\leq \frac{1}{2}+\frac{x_2-x_1}{2a}\leq\frac12,\\
    &x_2-x_1\leq a,\\
    &\frac{1-x_2-a}{1-a}\leq  \frac{1}{2}+\frac{x_2-x_1}{2a}\leq\frac{x_1}{1-a}.
    \end{align*}
    Solving the above optimization problem yields $x_1=x_2=(1-a)/2$ and $s_1=\frac12$.
\end{proof}

Combining Proposition \ref{pro:opt} and Proposition \ref{pro:ana}, we obtain:

\begin{proposition}\label{pro:poa}
$$PoA(a,\theta)=\begin{cases}
\displaystyle\frac{\theta-(\frac18-\frac{a}{2})}{\theta-(1-a)^2/4}& \text{ if }a\leq \frac14,\\
~~\\
\displaystyle\frac{\theta-(\frac14-a)}{\theta-(1-a)^2/4}& \text{ if }\frac14<a\leq 1.
\end{cases}$$
\end{proposition}

\begin{figure}[h]
\centering
\includegraphics[scale=0.7]{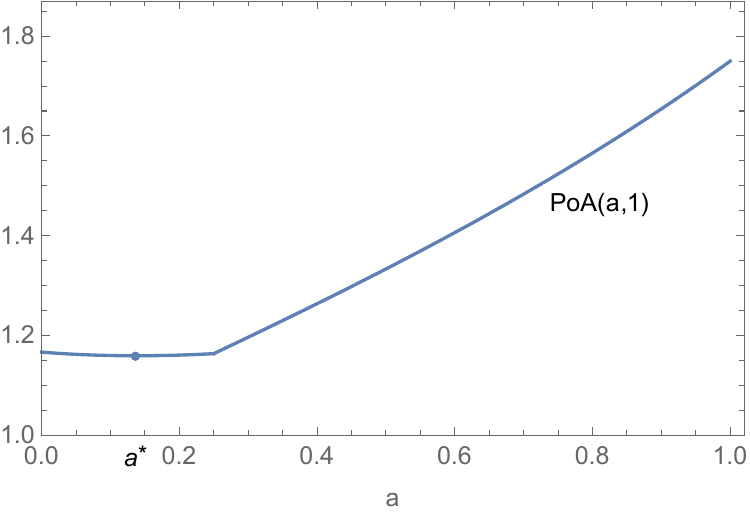}
\caption{\label{fig:poa2} Price of anarchy as a function of $a$ for $\theta=1$.}
\end{figure}

We find that $PoA(a,\theta)$ is decreasing with $a$ until $a^*=\frac{1-8\theta+\sqrt{64\theta^2-16\theta+9}}{4}$, then increasing with $a$. Therefore, the PoA is minimized for $a=a^*$. For $\theta=1$, we find $a^*=\frac{\sqrt{57}-7}{4} \simeq 0.137$ and $PoA(1,a^*) \simeq 1.159$.

\subsubsection{Price of stability}

\begin{proposition}\label{pro:sta}
A NE that maximizes the social welfare is as follows:\\
~~\\
If $a\leq \frac{2-\sqrt{2}}{2}$, then $(x_1,x_2)=\left(\frac{1-a}{2},\frac{1+a}{2}\right)$ and $s_1=\frac12$,  with a welfare of $\theta-(1-4a+2a^2)$.\\
~~\\
If $ \frac{2-\sqrt{2}}{2} \leq a\leq \frac12$, then $(x_1,x_2)=\left(\frac12-a,\frac12\right)$ and $s_1=\frac{1-2a}{2(1-a)}$,  with a welfare of $\theta-\frac{1 - 6 a + 12 a^2 - 12 a^3 + 4 a^4}{4 (1-a)^2}$.\\
~~\\
If $a>\frac12$, then $(x_1,x_2)=\left(0,\frac12\right)$ and $s_1=0$, with a welfare of $\theta - \frac14 + a$.
\end{proposition}
\begin{proof}
    Given that there are three different cases in Proposition \ref{prop:NE} and five different market equilibria (Proposition \ref{pro:vote}), we have to consider 15 cases. We illustrate the two equilibria that yield the highest welfare. The other 13 cases can be analyzed similarly.

    Case 1. Assume that $x_1\leq \frac12\leq x_2$ and consider market equilibrium $(iii)$. Then we want to maximize the welfare while guaranteeing that $(x_1,x_2)$ and $(s_1,1-s_1)$ defined by market equilibrium $(iii)$ is an NE, that is, we want to solve
    \begin{align*}
    \displaystyle\max_{s_1,x_1,x_2}\quad&\theta+ a s_1^2+a(1-s_1)^2 - \int_0^{x_1}(x_1-t)dt-\int_{x_1}^{s_1}(t-x_1)dt-\int_{s_1}^{x_2} (x_2-t)dt-\int_{x_2}^1 (t-x_2)dt  \\
    &0\leq x_1\leq \frac12\leq x_2\leq 1,\\
    &x_1\leq \frac{x_1+x_2-a}{2(1-a)}\leq x_2,\\
    &x_2-x_1\leq a,\\
    &\frac{x_2-a}{1-a}\leq  \frac{x_1+x_2-a}{2(1-a)}\leq\frac{x_1}{1-a}.
    \end{align*}
    Solving the above optimization problem yields $x_1=(1-a)/2$, $x_2=(1+a)/2$ and $s_1=\frac12$, which is the optimal solution for $a\leq (2-\sqrt{2})/2$.

    Case 2. Assume that $x_1\leq x_2\leq\frac12$ and consider market equilibrium $(iii)$. Then we want to maximize the welfare while guaranteeing that $(x_1,x_2)$ and $(s_1,1-s_1)$ defined by market equilibrium $(iii)$ is an NE, that is, we want to solve
    \begin{align*}
    \max\quad&\theta+ a s_1^2+a(1-s_1)^2 - \int_0^{x_1}(x_1-t)dt-\int_{x_1}^{s_1}(t-x_1)dt-\int_{s_1}^{x_2} (x_2-t)dt-\int_{x_2}^1 (t-x_2)dt  \\
    &0\leq x_1\leq \frac{x_1+x_2-a}{2(1-a)}\leq x_2\leq \frac12,\\
    &x_2-x_1\leq a,\\
    &\frac{1-x_2-a}{1-a}\leq  \frac{x_1+x_2-a}{2(1-a)}\leq\frac{x_1}{1-a}.
    \end{align*}
    Solving the above optimization problem yields $x_1=\frac12-a$, $x_2=\frac12$ and $s_1=(1-2a)/(2(1-a))$, which is the optimal solution for $(2-\sqrt{2})/2\leq a\leq \frac12$.
\end{proof}

Combining Proposition \ref{pro:opt} and Proposition \ref{pro:sta}, we obtain:

\begin{proposition} \label{pro:pos}
$$PoS(a,\theta)=\begin{cases}
\displaystyle\frac{\theta-(\frac18-a/2)}{\theta-(1-4a+2a^2)/4}& \text{ if }a\leq \frac14,\\
~~\\
\displaystyle\frac{\theta-(\frac14-a)}{\theta-(1-4a+2a^2)/4}& \text{ if }\frac14<a\leq (2-\sqrt{2})/2,\\
~~\\
\displaystyle\frac{\theta-(\frac14-a)}{\theta-(1-4a+2a^2)(1-2a+2a^2)/(4 (1-a)^2)}& \text{ if }(2-\sqrt{2})/2< a\leq \frac12,\\
~~\\
\displaystyle 1&\text{ if }a> \frac12
\end{cases}$$
\end{proposition}

\begin{figure}[h]
\centering
\includegraphics[scale=0.7]{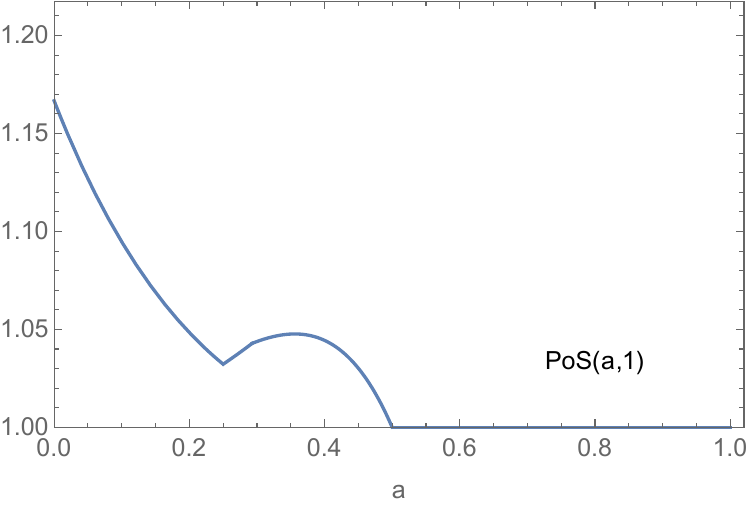}
\caption{\label{fig:pos} Price of stability as a function of $a$ for $\theta=1$.}
\end{figure}

When considering the best-case welfare analysis (price of stability), we observe that the inefficiency is not monotonic with respect to the parameter $a$. While efficiency is obtained for $a \geq \frac12$, it is not the case for $a < \frac12$.

\section{Appendix}\label{se:appendix}

\begin{proof}[Proof of Claim 1 in Theorem \ref{prop:NE}.]
Indeed, if $x_j\leq \frac12$ and $x_j+a\geq 1$, then $|x'_i-x_j|\leq a$ for all $0\leq x'_i\leq 1$. Suppose candidate $i$ deviates to $x'_i$. Then by Proposition \ref{pro:vote}, there is a market equilibrium with $s_i=0$ and thus such a deviation is not profitable.
So assume that $x_j\leq \frac12$ and $x_j+a< 1$. Suppose candidate $i$ deviates to $x'_i$. If $|x'_i-x_j|\leq a$, then by Proposition \ref{pro:vote}, there is a market equilibrium with $s_i=0$  and thus such a deviation is not profitable.
If $x'_i<x_j-a$ and $x_j-a>0$, then from Proposition \ref{pro:vote} the market equilibrium is unique and $s'_i=\frac{x'_i+x_j-a}{2(1-a)}$, which increases in $x'_i$. So the most beneficial location is when candidate $i$ locates at $x_j-a$ with $s'_i=\frac{x_j-a}{1-a}$.
If $x'_i>x_j+a$ and $x_j+a<1$, then by symmetry a similar argument implies that the most beneficial location is when candidate $i$ locates at $x_j+a$ with $s'_i=1-\frac{x_j}{1-a}$. Since $x_j\leq \frac12$, this deviation is more beneficial than locating at $x_j-a$.
\end{proof}

\begin{proof}[Proof of Proposition \ref{pro:ex_frac12}.] Case by case analysis when $a=\frac12$.\\
We first investigate equilibria where $x_1,x_2 \leq \frac12$.\\
$(i)$: $s_1=0$ is an NE if and only if $x_2-x_1\leq \frac12$ and $0 \in [1-2x_2,2x_1]$, therefore if and only if $x_1 \in [0,\frac12]$ and $x_2 = \frac12$.\\
$(ii)$: $s_1=\frac12-x_2+x_1$ is an NE if and only if $x_2-x_1\leq \frac12$, $\frac12\leq x_2$ and $\frac12-x_2+x_1 \in [1-2x_2,2x_1]$, therefore if and only if $x_1 \in [0,\frac12]$ and $x_2 = \frac12$.\\
$(iii)$: $s_1=x_1+x_2-\frac12$ is an NE if and only if $x_2-x_1\leq \frac12$, $x_1\leq \frac12\leq x_2$ and $x_1+x_2-\frac12 \in [1-2x_2,2x_1]$, therefore if and only if $x_1 \in [0,\frac12]$ and $x_2 = \frac12$.\\
$(iv)$: $s_1=\frac12-x_1+x_2$ is an NE if and only if $x_2-x_1\leq \frac12$, $x_1\leq \frac12$ and $\frac12-x_1+x_2 \in [1-2x_2,2x_1]$, therefore if and only if $x_1 \in [\frac14,\frac13]$ and $x_1\leq x_2 \leq 3 x_1 - \frac12$, and $x_1 \in [\frac13,\frac12]$ and $x_1\leq x_2 \leq \frac12$.\\
$(v)$: $s_1=1$ is a market equilibrium if and only if $x_2-x_1\leq \frac12$ and $1 \in [1-2x_2,2x_1]$, therefore if and only if $x_1 =x_2= \frac12$. 

We now investigate equilibrium where $x_1 \leq \frac12 \leq x_2$.\\
$(i)$: $s_1=0$ is an NE if and only if $x_2-x_1\leq \frac12$ and $0 \in [2x_2-1,2x_1]$, therefore if and only if $x_1 \in [0,\frac12]$ and $x_2 = \frac12$.\\
$(ii)$: $s_1=\frac12-x_2+x_1$ is an NE if and only if $x_2-x_1\leq \frac12$, $\frac12\leq x_2$ and $\frac12-x_2+x_1 \in [2x_2-1,2x_1]$, therefore if and only if $x_1\in[0,\frac12]$ and $x_2 \leq \frac{1}{2}+\frac{x_1}{3}$.\\
$(iii)$: $s_1=x_1+x_2-\frac12$ is an NE if and only if $x_2-x_1\leq \frac12$, $x_1\leq \frac12\leq x_2$ and $x_1+x_2-\frac12 \in [2x_2-1,2x_1]$, therefore if and only if $x_1\in[0,\frac12]$ and $x_2\leq x_1+ \frac12$.\\
$(iv)$: $s_1=\frac12-x_1+x_2$ is an NE if and only if $x_2-x_1\leq \frac12$, $x_1\leq \frac12$ and $\frac12-x_1+x_2 \in [2x_2-1,2x_1]$, therefore if and only if $x_1\in[\frac13,\frac12]$ and $x_2 \leq 3 x_1 - \frac12$.\\
$(v)$: $s_1=1$ is a market equilibrium if and only if $x_2-x_1\leq \frac12$ and $1 \in [2x_2-1,2x_1]$, therefore if and only if $x_1=\frac12$ and $x_2 \in [\frac12,1]$.

We lastly investigate equilibrium where $\frac12\leq x_1 \leq x_2$.\\
$(i)$: $s_1=0$ is an NE if and only if $x_2-x_1\leq \frac12$ and $0 \in [2x_2-1,2-2x_1]$, therefore if and only if $x_1=x_2 = \frac12$.\\
$(ii)$: $s_1=\frac12-x_2+x_1$ is an NE if and only if $x_2-x_1\leq \frac12$, $\frac12\leq x_2$ and $\frac12-x_2+x_1 \in [2x_2-1,2-2x_1]$, therefore if and only if $x_1\in[\frac12,\frac34]$ and $x_1\leq x_2 \leq \frac{1}{2}+\frac{x_1}{3}$.\\
$(iii)$: $s_1=x_1+x_2-\frac12$ is an NE if and only if $x_2-x_1\leq \frac12$, $x_1\leq \frac12\leq x_2$ and $x_1+x_2-\frac12 \in [2x_2-1,2-2x_1]$, therefore if and only if $x_1=\frac12$ and $x_2\in[\frac12,1]$.\\
$(iv)$: $s_1=\frac12-x_1+x_2$ is an NE if and only if $x_2-x_1\leq \frac12$, $x_1\leq \frac12$ and $\frac12-x_1+x_2 \in [2x_2-1,2-2x_1]$, therefore if and only if $x_1=\frac12$ and $x_2\in[\frac12,1]$.\\
$(v)$: $s_1=1$ is a market equilibrium if and only if $x_2-x_1\leq \frac12$ and $1 \in [2x_2-1,2-2x_1]$, therefore if and only if $x_1=\frac12$ and $x_2\in[\frac12,1]$.
\end{proof}

\bibliographystyle{plainnat}
\bibliography{popular}
\end{document}